\documentclass[conference,letterpaper]{IEEEtran}
\IEEEoverridecommandlockouts

\usepackage{cite}
\usepackage{url}
\usepackage{ifthen}
\usepackage{algorithm, algorithmicx, algpseudocode}
\usepackage{graphicx} 
\usepackage{textcomp}
\usepackage{booktabs}
\usepackage[cmex10]{amsmath}
\usepackage{amssymb,amsfonts}
\usepackage[left=1.62cm,right=1.62cm,top=1.85cm]{geometry}
\usepackage{pifont}
\usepackage{amsthm}
\usepackage{mathrsfs}
\def\BibTeX{{\rm B\kern-.05em{\sc i\kern-.025em b}\kern-.08em T\kern-.1667em\lower.7ex\hbox{E}\kern-.125emX}}
\usepackage{tikz}
\usetikzlibrary{automata,positioning,calc}
\usetikzlibrary{intersections}
\usetikzlibrary{decorations.pathreplacing,angles,quotes}

\usepackage{bm}
\usepackage{amscd}

\algnewcommand{\Initialize}[1]{%
  \State \textbf{Initialization:}
  \Statex \hspace*{\algorithmicindent}\parbox[t]{0.8\linewidth}{\raggedright #1}
}
\makeatletter

\theoremstyle{definition}
\newtheorem{theorem}{Theorem}
\newtheorem{definition}{Definition}
\newtheorem{lemma}{Lemma}
\newtheorem{prop}{Proposition}

\newtheorem{remark}{Remark}

\newcommand{\norm}[1]{\left\lVert#1\right\rVert}

\newcommand{\relmiddle}[1]{\mathrel{}\middle#1\mathrel{}} 

\def\br{\mathbb R}

\def\vE{\mathbb E}

\font\b=cmr10 scaled\magstep4

\def\bigzerou{\smash{\lower1.7ex\hbox{\b 0}}}
\def\bigzerou{\smash{\lower1.7ex\hbox{\b 0}}}
 
\begin{document}

\title{Several Representations of $\alpha$-Mutual Information and Interpretations as Privacy Leakage Measures
\thanks{This work was supported by JSPS KAKENHI Grant Number JP23K16886.}
}

\author{
\IEEEauthorblockN{Akira Kamatsuka}
\IEEEauthorblockA{Shonan Institute of Technology \\ 
Email: \text{kamatsuka@info.shonan-it.ac.jp}
 }
\and
\IEEEauthorblockN{Takahiro Yoshida}
\IEEEauthorblockA{Nihon University \\ 
Email: \text{yoshida.takahiro@nihon-u.ac.jp}
 } 
}

\maketitle

\begin{abstract}
In this paper, we present several novel representations of $\alpha$-mutual information ($\alpha$-MI) in terms of R{\'e}nyi divergence and conditional R{\' e}nyi entropy.
The representations are based on the variational characterizations of $\alpha$-MI using a reverse channel. 
Based on these representations, we provide several interpretations of the $\alpha$-MI as privacy leakage measures using generalized mean and gain functions.
Further, as byproducts of the representations, we propose novel conditional R{\' e}nyi entropies that satisfy the property that conditioning reduces entropy and data-processing inequality.
\end{abstract}

\section{Introduction} \label{sec:intro}
$\alpha$-Mutual information ($\alpha$-MI \cite{7308959}) is the common generalization of Shannon MI using a tunable parameter $\alpha \in (0, 1) \cup (1, \infty)$ and is closely related to various problems in information theory, such as the generalized cutoff rate, the error and correct decoding exponents in channel coding \cite{370121,e23020199}, 
the error and strong converse exponents in hypothesis testing \cite{6034266,8231191,8007073}, and the privacy-preserving data publishing \cite{8804205,8943950,9162168,10352344,6266165,6957119}. 

The $\alpha$-MI includes Sibson MI \cite{Sibson1969InformationR}, Arimoto MI \cite{arimoto1977}, Augustin--Csisz{\' a}r MI \cite{augusting_phd_thesis},\cite{370121}, Hayashi MI \cite{cryptoeprint:2013/440}, and Lapidoth--Pfister MI \cite{e21080778},\cite{8231191}. 
These types of $\alpha$-MI are defined as analogies of several representations of the corresponding Shannon MI based on R{\' e}nyi information, such as R{\' e}nyi entropy $H_{\alpha}(X)$, Arimoto conditional entropy $H_{\alpha}^{\text{A}}(X | Y)$, Hayashi conditional entropy $H_{\alpha}^{\text{H}}(X | Y)$ \cite{5773033}, and R{\' e}nyi divergence $D_{\alpha}(\cdot || \cdot)$. The Arimoto MI and Hayashi MI are defined as the difference between the R{\' e}nyi entropy and conditional R{\' e}nyi entropy: $I_{\alpha}^{(\cdot)}(X; Y):= H_{\alpha}(X) - H_{\alpha}^{(\cdot)}(X | Y)$. 
Meanwhile, the other types of $\alpha$-MI are defined as minimization problems concerning R{\' e}nyi divergence $D_{\alpha}(\cdot || \cdot)$.
Extensive theoretical properties and operational interpretations of the $\alpha$-MI and R{\'e}nyi information have been reported in the literature \cite{1055640,7282554,e21100969,e22050526,Nakiboglu:2019aa,8423117,8849809,e23060702,4595361,10619200,10619378,10619657,kamatsuka2024algorithms,10619174,10206941,10619672,8550766,9505206,9064819}. 

Recently, there has been extensive research on the relationship between $\alpha$-MI and privacy leakage measures in privacy-preserving data publishing \cite{8804205,8943950,9162168,10352344,6266165,6957119}.
For example, Liao \textit{et al.} \cite{8804205} showed that the Arimoto MI can be interpreted as the \textit{$\alpha$-leakage}, a privacy metric defined as the multiplicative increase of the maximal expected gain of a guessing adversary upon observing published data $Y$. 
Recently, Zarrabian and Sadeghi \cite{zarrabian2025extensionadversarialthreatmodel} provided an interpretation of Sibson MI as the generalized average of pointwise $\alpha$-leakage.
However, the interpretation of other types of $\alpha$-MI as privacy measures remains unknown.

In this paper, we investigate several representations and interpretations of $\alpha$-MI. Specifically, we address the following three issues:
\begin{description}
\item[Q$1$.] Can Arimoto MI and Hayashi MI be expressed in terms of R{\' e}nyi divergence?
\item[Q$2$.] Can Sibson MI, Augustin--Csisz{\' a}r MI, and Lapidoth--Pfister MI be expressed as the difference between R{\'e}nyi entropy and conditional entropy?
\item[Q$3$.] Can every type of $\alpha$-MI be interpreted as leakage measures based on the adversary's decision-making in privacy-preserving data publishing?
\end{description}

The remainder of this paper is organized as follows. 
In Section \ref{sec:preliminaries}, we review $\alpha$-MI.  
In Section \ref{ssec:diff_MI}, we address Q1 and Q2 by representing the Arimoto MI and Hayashi MI in terms of R{\' e}nyi divergence and by representing the Sibson MI, Augustin--Csisz{\' a}r MI, and Lapidoth--Pfister MI differentially using conditional R{\' e}nyi entropies. 
As byproducts of the latter representations, we propose novel conditional R{\' e}nyi entropies that satisfy the property that conditioning reduces entropy (CRE) and data-processing inequality (DPI) in Section \ref{ssec:novel_reneyi}. 
Section \ref{ssec:g_leakage} presents novel interpretations of $\alpha$-MI as privacy leakage measures, addressing Q3.

\section{Preliminaries}\label{sec:preliminaries}

Let $X$ and $Y$ be random variables on finite alphabets $\mathcal{X}$ and $\mathcal{Y}$, drawn according to a given joint distribution $p_{X, Y} = p_{X}p_{Y\mid X}$.
Let $H(X):=-\sum_{x}p_{X}(x)\log p_{X}(x)$, $H(X | Y):=-\sum_{x,y}p_{X}(x)p_{Y\mid X}(y | x)\log p_{X\mid Y}(x | y)$, 
and $I(X; Y):= H(X) - H(X | Y)$ be the Shannon entropy, conditional entropy, and Shannon MI, respectively.  
Let $\Delta_{\mathcal{X}}$ be a set of all probability distributions on $\mathcal{X}$. 
For $\alpha>0$ and a probability distribution $p\in \Delta_{\mathcal{X}}$, we denote $p^{(\alpha)}\in \Delta_{\mathcal{X}}$ the $\alpha$-tilted distribution \cite{8804205} (also known as escort distribution \cite{10.5555/3019383}) of $p$, defined as follows:
\begin{align}
p^{(\alpha)}(x) := \frac{p(x)^{\alpha}}{\sum_{x^{\prime}}p(x^{\prime})^{\alpha}}. \label{eq:alpha_tilted}
\end{align} 
Let $\vE^{p_{X}}[f(X)]:=\sum_{x}p_{X}(x)f(x)$ and $\mathbb{G}^{p_{X}}[f(X)] := \prod_{x}f(x)^{p_{X}(x)} = \exp\{\vE^{p_{X}}[\log f(X)]\}$ 
be the expectation and the geometric mean of $f(X)$, respectively.
In this study, we use $\log$ to represent the natural logarithm.

We first review the R{\' e}nyi entropy, R{\' e}nyi divergence, and $\alpha$-MI. 

\begin{definition}
Let $\alpha\in (0, 1)\cup (1, \infty)$. Given distributions $p_{X}$ and $q_{X}$, the R{\' e}nyi entropy of order $\alpha$, 
denoted by $H_{\alpha}(p_{X}) = H_{\alpha}(X)$, and the R{\' e}nyi divergence between $p_{X}$ and $q_{X}$ of order $\alpha$, denoted by $D_{\alpha}(p_{X} || q_{X})$, 
are defined as follows:
\begin{align}
H_{\alpha}(X) &:= \frac{1}{1-\alpha} \log \sum_{x} p_{X}(x)^{\alpha}, \\
D_{\alpha}(p_{X} || q_{X}) &:= \frac{1}{\alpha-1}\log \sum_{x} p_{X}(x)^{\alpha}q_{X}(x)^{1-\alpha}.
\end{align}
\end{definition}

\begin{definition}
Let $\alpha\in (0, 1)\cup (1, \infty)$ and $(X, Y)\sim p_{X, Y}=p_{X}p_{Y\mid X}$. 
The \textit{Sibson MI, Arimoto MI, Augustin--Csisz{\' a}r MI, Hayashi MI, and Lapidoth--Pfister MI of order $\alpha$}, denoted by $I_{\alpha}^{\text{S}}(X; Y), I_{\alpha}^{\text{A}}(X; Y), I_{\alpha}^{\text{C}}(X; Y), I_{\alpha}^{\text{H}}(X; Y)$, and $I_{\alpha}^{\text{LP}}(X; Y)$, respectively, 
are defined as follows:

\begin{align}
I_{\alpha}^{\text{S}} (X; Y) &:= \min_{q_{Y}} D_{\alpha} (p_{X}p_{Y\mid X} || p_{X}q_{Y}) \label{eq:Sibson_MI} \\ 
I_{\alpha}^{\text{A}}(X; Y) &:= H_{\alpha}(X) - H_{\alpha}^{\text{A}}(X\mid Y), \label{eq:Arimoto_MI} \\
I_{\alpha}^{\text{C}}(X; Y) &:= \min_{q_{Y}}\vE^{p_{X}}\left[D_{\alpha}(p_{Y\mid X}(\cdot \mid X) || q_{Y})\right], \label{eq:AC_MI}\\ 
I_{\alpha}^{\text{H}}(X; Y) &:= H_{\alpha}(X) - H_{\alpha}^{\text{H}}(X\mid Y), \label{eq:Hayshi_MI} \\
I_{\alpha}^{\text{LP}}(X; Y) &:= \min_{q_{X}}\min_{q_{Y}} D_{\alpha}(p_{X}p_{Y\mid X} || q_{X}q_{Y}), \label{eq:LP_MI}
\end{align}
where the minimum in \eqref{eq:Sibson_MI} and \eqref{eq:AC_MI} is taken over all possible probability distributions on $\mathcal{Y}$,  
the minimum in \eqref{eq:LP_MI} is taken over all possible product distributions on $\mathcal{X}\times \mathcal{Y}$, and 
$H_{\alpha}^{\text{A}}(X | Y):= \frac{\alpha}{1-\alpha}\log\sum_{y} \left( \sum_{x}p_{X}(x)^{\alpha}p_{Y\mid X}(y\mid x)^{\alpha} \right)^{\frac{1}{\alpha}}$ is Arimoto conditional entropy of order $\alpha$ \cite{arimoto1977}, 
$H_{\alpha}^{\text{H}}(X | Y):= \frac{1}{1-\alpha}\log\sum_{y}p_{Y}(y) \sum_{x}p_{X\mid Y}(x | y)^{\alpha} $ is the Hayashi conditional entropy of order $\alpha$ \cite{5773033}. 
\end{definition}

\begin{remark}
Notably, the values of $\alpha$-MI are extended by continuity to $\alpha=1$ and $\alpha=\infty$. 
$\alpha=1$ corresponds to the Shannon MI, $I(X; Y)$.
\end{remark}

The Sibson MI and Arimoto MI have the following representations based on the Gallager error exponent function $E_{0}(\rho, p_{X}):= -\log \sum_{y}\left( \sum_{x}p_{X}(x)p_{Y\mid X}(y | x)^{\frac{1}{1+\rho}} \right)^{1+\rho}$ \cite{Gallager:1968:ITR:578869}.

\begin{prop}[\text{\cite[Eqs. (13) and (16)]{370121}\cite[Eq.(2)]{7308959}}] \label{prop:relationship_Arimoto_Sibson}
\begin{align}
I_{\alpha}^{\text{S}}(X; Y) &= \frac{\alpha}{1-\alpha} E_{0} \left( \frac{1}{\alpha}-1, p_{X} \right),  \label{eq:Sibson_MI_Gallager} \\ 
I_{\alpha}^{\text{A}}(X; Y) &= \frac{\alpha}{1-\alpha} E_{0} \left( \frac{1}{\alpha}-1, p_{X_{\alpha}} \right), \label{eq:Arimoto_MI_Gallager}
\end{align}
where 
and $p_{X_{\alpha}}:=p_{X}^{(\alpha)}$ denotes the $\alpha$-tilted distribution of $p_{X}$, defined in \eqref{eq:alpha_tilted}. 
\end{prop}

It is known that if the channel $p_{Y\mid X}$ is noiseless, i.e., $X=Y$, then $\alpha$-MI is equal to the R{\' e}nyi entropy of $X$ of some order
\cite[Eq.(23)]{370121},\cite[Lemma 11]{e21080778},\cite[Thm 2, Remark 2]{cryptoeprint:2013/440},\cite[Thm 4.2]{esposito2024sibsons}.
\begin{prop} \label{prop:X_equal_Y}
For $\alpha\in (0, 1) \cup (1, \infty)$, 
\begin{align}
I_{\alpha}^{\text{S}}(X; X) &= H_{\frac{1}{\alpha}}(X), \\ 
I_{\alpha}^{\text{A}}(X; X) &= I_{\alpha}^{\text{H}}(X; X) = H_{\alpha}(X), \\ 
I_{\alpha}^{\text{C}}(X; X) &= H(X). 
\end{align}
For $\alpha \in (1/2, 1) \cup (1, \infty)$, 
\begin{align}
I_{\alpha}^{\text{LP}}(X; X) &= H_{\frac{\alpha}{2\alpha-1}}(X).
\end{align}
\end{prop}

Now, suppose that an adversary guesses $X$ from disclosed information $Y$ using a randomized decision rule $r_{\hat{X}\mid Y}$ with the following gain function $g_{\alpha}(\cdot, \cdot)\colon \mathcal{X}\times \Delta_{\mathcal{X}}\to \br$, which we refer to as \textit{$\alpha$-score}.

\begin{definition}[$\alpha$-score] Let $\alpha \in (0, \infty]$. Then, the \textit{$\alpha$-score} is defined as follows:
\begin{align}
g_{\alpha}(x, r) &:= 
\begin{cases}
\log r(x)-1, & \alpha=1, \\
\frac{\alpha}{\alpha-1}r(x)^{1-\frac{1}{\alpha}}, & \alpha\in (0, 1)\cup (1, \infty), \\ 
r(x), & \alpha=\infty.
\end{cases} \label{eq:alpha_score}
\end{align}

\end{definition}
Liao \textit{et al.} \cite{8804205} proved that Arimoto MI can be interpreted as a privacy leakage measure. 
\begin{prop}[\text{\cite[Thm 1]{8804205}}]
Let $\alpha\in (0, 1) \cup (1, \infty)$. Then, 
\begin{align}
I_{\alpha}^{\text{A}}(X; Y) &= \frac{\alpha}{\alpha-1} \log \frac{\max_{r_{\hat{X}\mid Y}}\vE_{X, Y}\left[g_{\alpha}(X, r_{\hat{X}\mid Y}(\cdot\mid Y))\right]}{\max_{r_{\hat{X}}} \vE\left[g_{\alpha}(X, r_{\hat{X}})\right]}, \label{eq:Arimoto_MI_leakage}
\end{align}
where the RHS of \eqref{eq:Arimoto_MI_leakage} is referred to as the \textit{$\alpha$-leakage from $X$ to $Y$} \cite[Def 8]{9162168}.
\end{prop}

\section{Main Results}\label{sec:main_result}
In this section, we address Q1--Q3, which were introduced in Section \ref{sec:intro}. 

\subsection{Representations of $\alpha$-MI using R{\' e}nyi Divergence and Conditional Entropies} \label{ssec:diff_MI}
We start by showing that Arimoto MI and Hayashi MI can be represented in terms of R{\' e}nyi divergence $D_{\alpha}(\cdot || \cdot)$, addressing Q$1$. 
\begin{prop}
\begin{align}
I_{\alpha}^{\text{A}}(X; Y) &= \min_{q_{Y}} D_{\alpha}(p_{X_{\alpha}}p_{Y\mid X} || p_{X_{\alpha}}q_{Y}) \label{eq:Arimoto_MI_divergence_01} \\ 
&= \min_{q_{Y}} \left\{D_{\alpha}(p_{X}p_{Y\mid X} || u_{X}q_{Y}) - D_{\alpha}(p_{X} || u_{X})\right\}, \label{eq:Arimoto_MI_divergence_02} \\ 
I_{\alpha}^{\text{H}}(X; Y) &= D_{\alpha}(p_{X_{\alpha}}p_{Y\mid X} || p_{X_{\alpha}}p_{Y}) \label{eq:Hayashi_MI_divergence_01} \\ 
&= D_{\alpha}(p_{X}p_{Y\mid X} || u_{X}p_{Y}) - D_{\alpha}(p_{X} || u_{X}), \label{eq:Hayashi_MI_divergence_02}
\end{align}
where $u_{X}$ denotes the uniform distribution on $\mathcal{X}$ and $p_{Y}$ denotes the marginal distribution on $Y$. 
\end{prop}
\begin{proof}
Eq. \eqref{eq:Arimoto_MI_divergence_01} follows directly from Eqs. \eqref{eq:Sibson_MI}, \eqref{eq:Sibson_MI_Gallager}, and \eqref{eq:Arimoto_MI_Gallager}.
Eqs. \eqref{eq:Arimoto_MI_divergence_02}, \eqref{eq:Hayashi_MI_divergence_01}, and \eqref{eq:Hayashi_MI_divergence_02} are derived by simple algebraic manipulation.
\end{proof}

To address Q2, we present the differential representations of Sibson MI, Augustin--Csisz{\' a}r MI, and Lapidoth--Pfister MI using conditional R{\' e}nyi entropies, a result consistent with  Proposition \ref{prop:X_equal_Y}.

\begin{theorem} \label{thm:alpha_MI_diff}
For $\alpha\in (0, 1)\cup (1, \infty)$, 
\begin{align}
I_{\alpha}^{\text{S}}(X; Y) &= H_{\frac{1}{\alpha}}(X) - H_{\alpha}^{\text{S}}(X\mid Y), \label{eq:Sibson_diff_expression}\\ 
I_{\alpha}^{\text{C}}(X; Y) &= H(X) - H_{\alpha}^{\text{C}}(X\mid Y), \label{eq:AC_diff_expression}
\end{align}
where 
\begin{align}
&H_{\alpha}^{\text{S}}(X|Y) := \min_{r_{X\mid Y}}\frac{\alpha}{1-\alpha}\notag \\ 
&\times \log \sum_{x,y}p_{X_{\frac{1}{\alpha}}}(x)p_{Y\mid X}(y | x)r_{X\mid Y}(x | y)^{1-\frac{1}{\alpha}}, \label{eq:Sibson_cond_renyi_ent} \\ 
&H_{\alpha}^{\text{C}}(X|Y) := \min_{r_{X\mid Y}}\frac{\alpha}{1-\alpha} \notag \\ 
&\times \sum_{x}p_{X}(x)\log \sum_{y}p_{Y\mid X}(y|x)r_{X\mid Y}(x|y)^{1-\frac{1}{\alpha}}. \label{eq:AC_cond_renyi_ent}
\end{align}
For $\alpha\in (1/2, 1) \cup (1, \infty)$, 
\begin{align}
I_{\alpha}^{\text{LP}}(X; Y) &= H_{\frac{\alpha}{2\alpha-1}}(X) - H_{\alpha}^{\text{LP}}(X\mid Y), \label{eq:LP_cond_renyi_ent}
\end{align}
where 
\begin{align}
&H_{\alpha}^{\text{LP}}(X|Y) := \min_{r_{X\mid Y}}\frac{2\alpha-1}{1-\alpha} \notag \\ 
&\times \log \sum_{x}p_{X_{\frac{\alpha}{2\alpha-1}}}(x)\left( \sum_{y}p_{Y\mid X}(y|x)r_{X\mid Y}(x|y)^{1-\frac{1}{\alpha}} \right)^{\frac{\alpha}{2\alpha-1}}.
\end{align}

\end{theorem}
\begin{proof}
See Appendix \ref{proof:alpha_MI_diff}.
\end{proof}

\begin{remark} \label{rem:VC_Arimoto_Hayashi_ent}
Similarly, $H_{\alpha}^{\text{A}}(X|Y)$ and $H_{\alpha}^{\text{H}}(X|Y)$ can be represented using a reverse channel $r_{X\mid Y}$ as follows \cite[p.141]{BN01990060en},\cite[Eq.(51)]{AkiraKAMATSUKA20252024TAP0010}.
\begin{align}
&H_{\alpha}^{\text{A}}(X|Y) =\min_{r_{X\mid Y}} \frac{\alpha}{1-\alpha} \log \sum_{x,y}p_{X}(x)p_{Y\mid X}(y|x)r_{X\mid Y}(x|y)^{1-\frac{1}{\alpha}}, \label{eq:Arimoto_cont_renyi_ent} \\ 
&H_{\alpha}^{\text{H}}(X|Y) = \min_{r_{X\mid Y}} \frac{1}{1-\alpha}\notag \\ 
&\times \log \sum_{x,y} \left( \alpha r_{X\mid Y}(x\mid y)^{\alpha-1}-(\alpha-1)\norm{r_{X\mid Y}(\cdot \mid y)}_{\alpha}^{\alpha} \right). \label{eq:Hayashieq_cont_renyi_ent}
\end{align}
\end{remark}

\subsection{Novel conditional R{\' enyi} entropies} \label{ssec:novel_reneyi}
To the best of our knowledge, so far, $H_{\alpha}^{\text{A}}(X|Y)$ and $H_{\alpha}^{\text{H}}(X|Y)$ are the only conditional R{\' enyi} entropies that satisfy both CRE, i.e., $H_{\alpha}(X)\geq H_{\alpha}^{(\cdot)}(X|Y)$, and DPI (also known as monotonicity) (see \cite[Table 1]{cryptoeprint:2013/440},\cite[Table 1]{ecea-4-05030}).
From Eqs.\eqref{eq:Sibson_cond_renyi_ent}, \eqref{eq:AC_cond_renyi_ent}, and \eqref{eq:LP_cond_renyi_ent}, we obtain novel conditional R{\' e}nyi entropy-like quantities that satisfy CRE and DPI in the following sense. 

\begin{prop}
Let $\tilde{H}_{\alpha}^{\text{S}}(X|Y):=H_{\frac{1}{\alpha}}^{\text{S}}(X|Y), \tilde{H}_{\alpha}^{\text{LP}}(X|Y):=H_{\frac{\alpha}{2\alpha-1}}^{\text{LP}}(X|Y)$.
\begin{enumerate}
\item Then, the following holds.
\begin{align}
H_{\alpha}(X) &\geq \tilde{H}_{\alpha}^{\text{S}}(X|Y), \qquad \alpha \in (0, 1) \cup (1, \infty), \\ 
H_{\alpha}(X) &\geq \tilde{H}_{\alpha}^{\text{LP}}(X|Y), \quad \alpha \in (1/2, 1)\cup (1, \infty) , \\ 
H(X) &\geq H_{\alpha}^{\text{C}}(X|Y), \qquad \alpha \in (0, 1) \cup (1, \infty). 
\end{align}
\item If $X-Y-Z$ forms a Markov chain, then 
\begin{align}
\tilde{H}_{\alpha}^{\text{S}}(X|Y) &\leq \tilde{H}_{\alpha}^{\text{S}}(X|Z), \quad \alpha \in (0, 1) \cup (1, \infty), \\ 
\tilde{H}_{\alpha}^{\text{LP}}(X|Y) &\leq \tilde{H}_{\alpha}^{\text{LP}}(X|Z), \quad \hspace{-8pt} \alpha \in (1/2, 1) \cup (1, \infty) , \\ 
H_{\alpha}^{\text{C}}(X|Y) &\leq H_{\alpha}^{\text{C}}(X|Z), \quad \alpha \in (0, 1) \cup (1, \infty). 
\end{align}
\end{enumerate}
\end{prop}
\begin{proof}
It follows from the nonnegativity and DPI of Sibson MI, Augustin--Csisz{\' a}r MI, and Lapidoth--Pfister MI \cite[Thm 4.1]{esposito2024sibsons},\cite[Eq.(22)]{370121},\cite[Lemma 4]{e21080778}.
\end{proof}


\subsection{Interpretation of $\alpha$-MI as Privacy Leakage Measure} \label{ssec:g_leakage}
In this subsection, we address Q$3$, i.e., we provide novel interpretations of $\alpha$-MI as privacy leakage measures.
To this end, we first introduce the generalized mean and gain functions.

\begin{definition}[Power mean] \label{def:power_mean}
Let $t\in (-\infty, \infty)\setminus \{0\}$. Given a distribution $p_{X}$ and a nonnegative function $f$, the \textit{power mean} (also known as \textit{generalized mean} or \textit{H{\" o}lder mean}) \textit{of order $t$ of $f(X)$} is defined as follows.
\begin{align}
\mathbb{M}_{t}^{p_{X}}[f(X)] &:= \{\vE^{p_{X}}\left[f(X)^{t}\right]\}^{\frac{1}{t}}.
\end{align}
\end{definition}
\begin{remark}
The power mean recovers the expectation $\vE^{p_{X}}[\cdot]$ when $t = 1$ and converges to the geometric mean $\mathbb{G}^{p_{X}}[\cdot]$ as $t\to 0$. 
Thus, the power mean is extended by continuity to $t=0$, denoted by $\mathbb{M}_{0}^{p_{X}}[\cdot] := \mathbb{G}^{p_{X}}[\cdot]$.
\end{remark}

The power mean can be seen as the generalized geometric mean, which is defined as follows.

\begin{definition}[Generalized geometric mean] \label{def:gen_geometic_mean}
Let $q\in (-\infty, \infty)$. Given a distribution $p_{X}$, the \textit{generalized geometric mean of order $q$ of $f(X)$} is defined as follows:
\begin{align}
\mathbb{G}_{q}^{p_{X}}[f(X)] &:= \exp_{q}\left\{\vE^{p_{X}}\left[\ln_{q}f(X)\right]\right\}, 
\end{align}
where $\ln_{q} (\cdot)$ and $\exp_{q}\{\cdot\}$ are $q$-logarithm and $q$-exponential, respectively, defined as follows:
\begin{align}
\ln_{q}x &:= 
\begin{cases}
\log x, & q=1, \\ 
\frac{x^{1-q}-1}{1-q}, & q\neq 1, 
\end{cases} \\ 
\exp_{q}\{x\} &:= 
\begin{cases}
\exp\{x\}, & q=1, \\ 
[1+(1-q)x]^{\frac{1}{1-q}}, & q\neq 1.
\end{cases}
\end{align}
\end{definition}

\begin{lemma} \label{lem:equivalent_power_gen_geometic_mean}
Let $q\in (-\infty, \infty)$. Then, 
\begin{align}
\mathbb{M}_{1-q}^{p_{X}}[f(X)] = \mathbb{G}_{q}^{p_{X}}[f(X)].
\end{align}
\end{lemma}
\begin{proof}
\begin{align}
\mathbb{M}_{1-q}^{p_{X}}[f(X)] &:= \{\vE^{p_{X}}\left[f(X)^{1-q}\right]\}^{\frac{1}{1-q}} \\ 
&= \left\{1 + (1-q) \vE^{p_{X}}\left[\frac{f(X)^{1-q}-1}{1-q}\right]\right\}^{\frac{1}{1-q}} \\ 
&= \mathbb{G}_{q}^{p_{X}}[f(X)].
\end{align}
\end{proof}

To interpret all types of $\alpha$-MI in the context of privacy-preserving data publishing problems, we assume guessing adversaries with the following gain functions in addition to $\alpha$-score $g_{\alpha}(x,r)$ defined in \eqref{eq:alpha_score}.
\begin{definition}[Pseudospherical score, Power score]
Let $\alpha\in (0, 1)\cup (1, \infty)$. The \textit{pseudospherical score} $g_{\alpha, \text{PS}}(x, r)$ and \textit{power score} (also known as \textit{Tsallis score} \cite{Dawid:2014ua}) $g_{\alpha, \text{PW}}(x, r)$ are defined as follows.
\begin{align}
g_{\alpha, \text{PS}}(x, r) &= g_{\alpha}(x, r^{(\alpha)}) = \frac{\alpha}{\alpha-1}\cdot \left(\frac{r(x)}{\norm{r}_{\alpha}} \right)^{\alpha-1}, \\
g_{\alpha, \text{PW}}(x, r) &= \frac{\alpha}{\alpha-1}\cdot r(x)^{\alpha-1} - \norm{r}_{\alpha}^{\alpha}.
\end{align}
\end{definition}
\begin{remark}
These scores, as well as  $g_{1}(x, r)$, are known as the \textit{proper scoring rule} (PSR)\footnote{A gain function $g(x, r)$ is the proper scoring rule if the expected gain is maximized at $r=p_{X}$, i.e.,  for all $r\in \Delta_{\mathcal{X}}$, $\vE^{p_{X}}[g(X, p_{X})]\geq \vE^{p_{X}}[g(X, r)]$. } \cite{Dawid:2014ua,doi:10.1198/016214506000001437}. 
Therefore, the adversary that maximizes the expected gain of these gain functions estimates $X$ using the posterior distribution $r_{\hat{X}\mid Y} = p_{X\mid Y}$ for any $\alpha\in (0, 1)\cup (1, \infty)$.
Meanwhile, the adversary that maximizes the expected gain of $\alpha$-score estimates $X$ using the $\alpha$-tilted distribution of $p_{X\mid Y}$, i.e., $r_{\hat{X}\mid Y}(x|y) = \frac{p_{X\mid Y}(x|y)^{\alpha}}{\sum_{x}p_{X\mid Y}(x|y)^{\alpha}}$ \cite[Lemma 1]{8804205}.
\end{remark}

Now, we present the interpretations of $\alpha$-MI as privacy metrics. 
Every $\alpha$-MI can be interpreted as privacy metrics defined as the multiplicative increase of the maximal generalized mean of gain of a guessing adversary upon observing published data $Y$, as follows.
\begin{theorem} \label{thm:interpetations_alpha_MI}
For $\alpha \in (0, 1)\cup (1, \infty)$, 
\begin{align}
&I_{\alpha}^{\text{A}}(X; Y) \notag \\ 
&= \frac{\alpha}{\alpha-1} \log \frac{\max_{r_{\hat{X}\mid Y}}\vE^{p_{X}p_{Y\mid X}}[g_{\alpha}(X, r_{\hat{X}\mid Y}(\cdot \mid Y))]}{\max_{r_{\hat{X}}}\vE^{p_{X}}[g_{\alpha}(X, r_{\hat{X}})]} \label{eq:Arimoto_privacy_01} \\ 
&= \frac{\alpha}{\alpha-1} \log \frac{\max_{r_{\hat{X}\mid Y}}\vE^{p_{X}p_{Y\mid X}}[g_{\alpha, \text{PS}}(X, r_{\hat{X}\mid Y}(\cdot \mid Y))]}{\max_{r_{\hat{X}}}\vE^{p_{X}}[g_{\alpha, \text{PS}}(X, r_{\hat{X}})]} \label{eq:Arimoto_privacy_02} \\ 
&= \log \frac{\max_{r_{\hat{X}\mid Y}}\mathbb{M}_{1-\frac{1}{\alpha}}^{p_{X}p_{Y\mid X}}[g_{\infty}(X, r_{\hat{X}\mid Y}(\cdot\mid Y))]}{\max_{r_{\hat{X}}}\mathbb{M}_{1-\frac{1}{\alpha}}^{p_{X}}[g_{\infty}(X, r_{\hat{X}})]} \label{eq:Arimoto_privacy_03} \\ 
&= \log \frac{\max_{r_{\hat{X}\mid Y}}\mathbb{G}_{\frac{1}{\alpha}}^{p_{X}p_{Y\mid X}}[g_{\infty}(X, r_{\hat{X}\mid Y}(\cdot\mid Y))]}{\max_{r_{\hat{X}}}\mathbb{G}_{\frac{1}{\alpha}}^{p_{X}}[g_{\infty}(X, r_{\hat{X}})]}, \label{eq:Arimoto_privacy_04} \\ 
&I_{\alpha}^{\text{S}}(X; Y) \notag \\ 
&= \frac{\alpha}{\alpha-1} \log \frac{\max_{r_{\hat{X}\mid Y}}\vE^{p_{X_{\frac{1}{\alpha}}}p_{Y\mid X}}[g_{\alpha}(X, r_{\hat{X}\mid Y}(\cdot \mid Y))]}{\max_{r_{\hat{X}}}\vE^{p_{X_{\frac{1}{\alpha}}}}[g_{\alpha}(X, r_{\hat{X}})]} \label{eq:Sibson_privacy_01} \\ 
&= \frac{\alpha}{\alpha-1} \log \frac{\max_{r_{\hat{X}\mid Y}}\vE^{p_{X_{\frac{1}{\alpha}}}p_{Y\mid X}}[g_{\alpha, \text{PS}}(X, r_{\hat{X}\mid Y}(\cdot \mid Y))]}{\max_{r_{\hat{X}}}\vE^{p_{X_{\frac{1}{\alpha}}}}[g_{\alpha, \text{PS}}(X, r_{\hat{X}})]} \label{eq:Sibson_privacy_02} \\ 
&= \log \frac{\max_{r_{\hat{X}\mid Y}}\mathbb{M}_{1-\frac{1}{\alpha}}^{p_{X_{\frac{1}{\alpha}}}p_{Y\mid X}}[g_{\infty}(X, r_{\hat{X}\mid Y}(\cdot\mid Y))]}{\max_{r_{\hat{X}}}\mathbb{M}_{1-\frac{1}{\alpha}}^{p_{X_{\frac{1}{\alpha}}}}[g_{\infty}(X, r_{\hat{X}})]} \label{eq:Sibson_privacy_03} \\ 
&= \log \frac{\max_{r_{\hat{X}\mid Y}}\mathbb{G}_{\frac{1}{\alpha}}^{p_{X_{\frac{1}{\alpha}}}p_{Y\mid X}}[g_{\infty}(X, r_{\hat{X}\mid Y}(\cdot\mid Y))]}{\max_{r_{\hat{X}}}\mathbb{G}_{\frac{1}{\alpha}}^{p_{X_{\frac{1}{\alpha}}}}[g_{\infty}(X, r_{\hat{X}})]}, \label{eq:Sibson_privacy_04} \\ 
&I_{\alpha}^{\text{C}}(X; Y) \notag \\ 
&= \frac{\alpha}{\alpha-1} \log \frac{\max_{r_{\hat{X}\mid Y}}\mathbb{G}^{p_{X}}\left[\vE^{p_{Y\mid X}}\left[g_{\alpha}(X, r_{\hat{X}\mid Y}(\cdot \mid Y)) \relmiddle{|} X\right]\right]}{\max_{r_{\hat{X}}}\mathbb{G}^{p_{X}}[g_{\alpha}(X, r_{\hat{X}})]} \label{eq:AC_privacy_01} \\
&= \frac{\alpha}{\alpha-1} \log \frac{\max_{r_{\hat{X}\mid Y}}\mathbb{G}^{p_{X}}\left[\vE^{p_{Y\mid X}}\left[g_{\alpha, \text{PS}}(X, r_{\hat{X}\mid Y}(\cdot \mid Y)) \relmiddle{|} X\right]\right]}{\max_{r_{\hat{X}}}\mathbb{G}^{p_{X}}[g_{\alpha, \text{PS}}(X, r_{\hat{X}})]} \label{eq:AC_privacy_02} \\
&= \log \frac{\max_{r_{\hat{X}\mid Y}}\mathbb{M}_{0}^{p_{X}}\left[\mathbb{M}_{1-\frac{1}{\alpha}}^{p_{Y\mid X}}\left[g_{\infty}(X, r_{\hat{X}\mid Y}(\cdot \mid Y))\relmiddle{|}X\right]\right]}{\max_{r_{\hat{X}}}\mathbb{G}^{p_{X}}[g_{\infty}(X, r_{\hat{X}})]} \label{eq:AC_privacy_03} \\ 
&= \log \frac{\max_{r_{\hat{X}\mid Y}}\mathbb{G}^{p_{X}}\left[\mathbb{G}_{\frac{1}{\alpha}}^{p_{Y\mid X}}\left[g_{\infty}(X, r_{\hat{X}\mid Y}(\cdot \mid Y))\relmiddle{|}X\right]\right]}{\max_{r_{\hat{X}}}\mathbb{G}^{p_{X}}[g_{\infty}(X, r_{\hat{X}})]}, \label{eq:AC_privacy_04} \\ 
&I_{\alpha}^{\text{H}}(X; Y) \notag \\ 
&= \frac{1}{\alpha-1} \log \frac{\max_{r_{\hat{X}\mid Y}} \vE^{p_{X}p_{Y\mid X}}\left[g_{\alpha, \text{PW}}(X, r_{\hat{X}\mid Y}(\cdot \mid Y))\right]}{\max_{r_{\hat{X}}}\vE^{p_{X}}[g_{\alpha, \text{PW}}(X, r_{\hat{X}})]}. \label{eq:Hayashi_privacy_01}
\end{align}
For $\alpha \in (1/2, 1) \cup (1, \infty)$, 
\begin{align}
&I_{\alpha}^{\text{LP}}(X; Y) = \frac{\alpha}{\alpha-1} \notag \\ 
&\times \log  \frac{\max_{r_{\hat{X}\mid Y}}\mathbb{M}_{\frac{\alpha}{2\alpha-1}}^{p_{X_{\frac{\alpha}{2\alpha-1}}}}\left[\vE^{p_{Y\mid X}}\left[g_{\alpha}(X, r_{\hat{X}\mid Y}(\cdot | Y)) \relmiddle{|}X\right]\right]}{\max_{r_{\hat{X}}}\mathbb{M}_{\frac{\alpha}{2\alpha-1}}^{p_{X_{\frac{\alpha}{2\alpha-1}}}}[g_{\alpha}(X, r_{\hat{X}})]} \label{eq:LP_privacy_01} \\
&= \frac{\alpha}{\alpha-1} \notag \\ 
&\times \log  \frac{\max_{r_{\hat{X}\mid Y}}\mathbb{M}_{\frac{\alpha}{2\alpha-1}}^{p_{X_{\frac{\alpha}{2\alpha-1}}}}\left[\vE^{p_{Y\mid X}}\left[g_{\alpha, \text{PS}}(X, r_{\hat{X}\mid Y}(\cdot | Y)) \relmiddle{|}X\right]\right]}{\max_{r_{\hat{X}}}\mathbb{M}_{\frac{\alpha}{2\alpha-1}}^{p_{X_{\frac{\alpha}{2\alpha-1}}}}[g_{\alpha, \text{PS}}(X, r_{\hat{X}})]}. \label{eq:LP_privacy_02}
\end{align}
\end{theorem}
\begin{proof}
See Appendix \ref{proof:interpetations_alpha_MI}.
\end{proof}
\begin{remark}
Arimoto MI can be interpreted as the multiplicative increase of the maximal (generalized) expected gain of the $g_{\alpha}(x, r), g_{\alpha, \text{PS}}(x, r)$, and $g_{\infty}(x, r)$ of a guessing adversary upon observing published data $Y$, where Eq. \eqref{eq:Arimoto_privacy_01} was originally proved by Liao \textit{et al.} \cite[Thm 1]{8804205}. 
Sibson MI can also be interpreted as a privacy measure similar to Arimoto MI, but the distribution over which the mean is taken changes from $p_{X}$ to $p_{X_{\frac{1}{\alpha}}}$. 
This can be considered a privacy metric when assuming the following adversary\footnote{Similar observations on the tilted distribution are in \cite{10619672} and \cite{9761766}. }. 
\begin{itemize}
\item When $\alpha \in (0, 1)$, the adversary tends to overestimate the probability of large values of $p_{X}(x)$ and ignore the probability of small values of $p_{X}(x)$.  
\item When $\alpha \in (1, \infty)$, the adversary has a balanced belief on $X$ that is close to a uniform distribution overall by estimating high values of $p_{X}(x)$ as low and small values as large.
\end{itemize}
From \eqref{eq:Hayashi_privacy_01}, Hayashi MI can be interpreted as a privacy measure when assuming an adversary with a different gain function, $g_{\alpha, \text{PW}}(x, r)$, compared with Arimoto MI.
Meanwhile, Augustin--Csisz{\' a}r MI differs from Arimoto MI in that it takes the mean of $X$ as a regular geometric mean $\mathbb{G}^{p_{X}}[\cdot]$. 
Therefore, Augustin--Csisz{\' a}r MI can be interpreted as a privacy measure when an adversary wants to guess ratio-related data, such as growth or interest rates, and seeks to increase the average gain in the geometric mean rather than the usual expectation.
Finally, Lapidoth--Pfister MI differs from Arimoto MI in that it takes the mean for $X$ as a power mean of order $\frac{\alpha}{2\alpha-1}$ using a $\frac{\alpha}{2\alpha-1}$-tilted distribution of $p_{X}$, i.e., $\mathbb{M}_{\frac{\alpha}{2\alpha-1}}^{p_{X_{\frac{\alpha}{2\alpha-1}}}}[\cdot]$.
Notably, the power mean $\mathbb{M}_{\frac{\alpha}{2\alpha-1}}^{p_{X_{\frac{\alpha}{2\alpha-1}}}}[\cdot]$ transitions from $\vE^{p_{X}}[\cdot]$ to $\mathbb{M}_{\frac{1}{2}}^{p_{X_{\frac{1}{2}}}}[\cdot]$ when $\alpha \in (1, \infty)$, and in general, $\mathbb{G}^{p_{X}}[\cdot]\leq \mathbb{M}_{\frac{1}{2}}^{p_{X_{\frac{1}{2}}}}[\cdot] \leq \vE^{p_{X}}[\cdot]$.
This means that Lapidoth--Pfister MI can be considered a privacy measure for an adversary with intermediate properties between those assumed in Arimoto MI and Augustin--Csisz{\' a}r MI when $\alpha\in (1, \infty)$.
\end{remark}

\section{Conclusion}\label{sec:conclusion}
This paper provides several novel representations of various types of $\alpha$-MI using divergence and conditional entropy. 
Based on these representations, we propose novel R{\' e}nyi conditional entropies and interpretations of $\alpha$-MI as privacy metric. 
As byproducts of the representations, we propose novel conditional R{\' enyi} entropies that satisfy CRE and DPI.

\appendices

\section{Proof of Theorem \ref{thm:alpha_MI_diff}} \label{proof:alpha_MI_diff}

To prove Theorem \ref{thm:alpha_MI_diff}, we first state the following lemmas.

\begin{lemma} \label{lem:invariance_renyi_entropy}
Let $\alpha\in (0, 1) \cup (1, \infty)$. Given a distribution $p_{X}$, 
\begin{enumerate}
\item The $\alpha$-tilted distribution of the $\frac{1}{\alpha}$-tilted distribution of $p_{X}$ is $p_{X}$, i.e., 
$p_{(X_{\frac{1}{\alpha}})_{\alpha}} = p_{X}$. 
\item 
\begin{align}
H_{\alpha}(p_{X_{\frac{1}{\alpha}}}) &= H_{\frac{1}{\alpha}}(p_{X}).
\end{align}
\end{enumerate}
\end{lemma}
\begin{proof}
1) follows from \cite[Prop 1]{10619200}. We only prove 2) as follows.
\begin{align}
H_{\alpha}(p_{X_{\frac{1}{\alpha}}}) &= \frac{\alpha}{1-\alpha} \log \left( \sum_{x} \left( \frac{p_{X}(x)^{\frac{1}{\alpha}}}{\sum_{x}p_{X}(x)^{\frac{1}{\alpha}}} \right)^{\alpha} \right)^{\frac{1}{\alpha}} \\ 
&= \frac{1}{\alpha-1} \log \left( \sum_{x}p_{X}(x)^{\frac{1}{\alpha}} \right)^{\alpha} = H_{\frac{1}{\alpha}}(p_{X}).
\end{align}
\end{proof}

\begin{lemma} \label{lem:VC_AC_LP_MI}
For $\alpha \in (0,1)\cup(1, \infty)$, 
\begin{align}
&I_{\alpha}^{\text{C}}(X; Y) = H(p_{X}) \notag \\ 
&+ \max_{r_{X\mid Y}} \frac{\alpha}{\alpha-1}\sum_{x}p_{X}(x) \log \sum_{y}p_{Y\mid X}(y | x)  r_{X\mid Y}(x | y)^{1-\frac{1}{\alpha}}. \label{eq:VC_AC_MI} 
\end{align}
For $\alpha \in (1/2, 1)\cup (1, \infty)$, 
\begin{align}
&I_{\alpha}^{\text{LP}}(X; Y) = \max_{r_{X\mid Y}} \Big\{\frac{2\alpha-1}{\alpha-1}  \notag \\ 
&\times \log \sum_{x}p_{X}(x)^{\frac{\alpha}{2\alpha-1}} \left( \sum_{y} p_{Y\mid X}(y | x) r_{X\mid Y}(x | y)^{1-\frac{1}{\alpha}} \right)^{\frac{\alpha}{2\alpha-1}} \Big\}. \label{eq:VC_LP_MI}
\end{align}

\end{lemma}
\begin{proof}
Eq. \eqref{eq:VC_AC_MI} has already proven in \cite[Thm 2]{kamatsuka2024algorithms}; we reprove it here for completeness.

We define the distributions as follows: 
\begin{align}
&\hat{q}_{Y\mid X}(y\mid x) := \frac{p_{Y\mid X}(y\mid x)r_{X\mid Y}(x\mid y)^{1-\frac{1}{\alpha}}}{\sum_{y}p_{Y\mid X}(y\mid x)r_{X\mid Y}(x\mid y)^{1-\frac{1}{\alpha}}} \label{eq:hat_q_Y_X} \\ 
&\hat{q}_{X}(x) := \frac{p_{X}(x)^{\frac{\alpha}{2\alpha-1}} \left( \sum_{y}p_{Y\mid X}(y | x)r_{X\mid Y}(x | y)^{1-\frac{1}{\alpha}} \right)^{\frac{\alpha}{2\alpha-1}}}{\sum_{x}p_{X}(x)^{\frac{\alpha}{2\alpha-1}} \left( \sum_{y}p_{Y\mid X}(y | x)r_{X\mid Y}(x | y)^{1-\frac{1}{\alpha}} \right)^{\frac{\alpha}{2\alpha-1}}}, \label{eq:update_LP_q_X}
\end{align}
Then, for $\alpha \in (0, 1)$, 
\begin{align}
&I_{\alpha}^{\text{C}}(X; Y) \notag \\ 
&\overset{(a)}{=} \min_{\tilde{q}_{Y\mid X}} \left\{I(p_{X}, \tilde{q}_{Y\mid X}) + \frac{\alpha}{1-\alpha} D(p_{X}\tilde{q}_{Y\mid X} || p_{X}p_{Y\mid X})\right\} \label{eq:proof_begin} \\ 
&\overset{(b)}{=} \min_{\tilde{q}_{Y\mid X}} \max_{r_{X\mid Y}} 
\Biggl\{\vE^{p_{X}\tilde{q}_{Y\mid X}} \left[\log \frac{r_{X\mid Y}(X\mid Y)}{p_{X}(X)}\right] \notag \\ 
&\qquad \qquad \qquad \qquad + \frac{\alpha}{1-\alpha} D(p_{X}\tilde{q}_{Y\mid X} || p_{X}p_{Y\mid X})\Biggr\} \\ 
&\overset{(c)}{=} \max_{r_{X\mid Y}} \min_{\tilde{q}_{Y\mid X}}
\Biggl\{\vE^{p_{X}\tilde{q}_{Y\mid X}} \left[\log \frac{r_{X\mid Y}(X\mid Y)}{p_{X}(X)}\right] \notag \\ 
&\qquad \qquad \qquad \qquad + \frac{\alpha}{1-\alpha} D(p_{X}\tilde{q}_{Y\mid X} || p_{X}p_{Y\mid X})\Biggr\} \\
&= \max_{r_{X\mid Y}} \min_{\tilde{q}_{Y\mid X}} \Biggl\{H(p_{X}) \notag \\\ 
&+\frac{\alpha}{1-\alpha} \vE^{p_{X}} \left[D\left(\tilde{q}_{Y\mid X}(\cdot \mid X) || \hat{q}_{Y\mid X}(\cdot \mid X)\right)\right] 
\notag \\
&+ \frac{\alpha}{\alpha-1} \sum_{x}p_{X}(x)\log \sum_{y} p_{Y\mid X}(y | x)r_{X\mid Y}(x | y)^{1-\frac{1}{\alpha}}\Biggr\} \label{eq:proof_end} \\
&\overset{(d)}{=} \eqref{eq:VC_AC_MI}, 
\end{align}
where $\hat{q}_{Y\mid X}(\cdot | x)$ is defined in \eqref{eq:hat_q_Y_X} and 
$(a)$ follows from \cite[Thm 1]{6034266}, $(b)$ follows from a variational representation of the Shannon MI using a reverse channel (see, e.g., \cite[Lemma 10.8.1]{Cover:2006:EIT:1146355}), 
$(c)$ follows from the minimax theorem (see \cite[Thm 4.2]{pjm/1103040253})\footnote{Notably, $\tilde{F}_{\alpha}^{\text{C}}(\tilde{q}_{Y\mid X}, r_{X\mid Y}):= \vE^{p_{X}\tilde{q}_{Y\mid X}} \left[\log \frac{r_{X\mid Y}(X\mid Y)}{p_{X}(X)}\right] + \frac{\alpha}{1-\alpha} D(p_{X}\tilde{q}_{Y\mid X} || p_{X}p_{Y\mid X}), \alpha \in (0, 1)$ is concave with respect to $r_{X\mid Y}$ (see, e.g. \cite[Sec. 10.3.2]{10.5555/1199866}) and is linear (hence convex) with respect to $\tilde{q}_{Y\mid X}$. },  
and $(d)$ follows from the properties of the Kullback--Leibler divergence (see, e.g., \cite[Thm 2.6.3]{Cover:2006:EIT:1146355}). 
For $\alpha \in (1, \infty)$, Eq. \eqref{eq:VC_AC_MI} can be obtained by replacing $\min_{\tilde{q}_{Y\mid X}}$ with $\max_{\tilde{q}_{Y\mid X}}$ in Eqs. \eqref{eq:proof_begin}--\eqref{eq:proof_end} (see \cite[Thm 1]{6034266}).

Similarly, for $\alpha \in (1/2, 1)$, 
\begin{align}
&I_{\alpha}^{\text{LP}}(X; Y) 
\overset{(e)}{=} \min_{\tilde{q}_{X, Y}}\left\{ \frac{\alpha}{1-\alpha} D(\tilde{q}_{X, Y} || p_{X}p_{Y\mid X}) + I(\tilde{q}_{X}, \tilde{q}_{Y\mid X}) \right\} \label{eq:proof_begin_LP} \\
&\overset{(f)}{=} \min_{\tilde{q}_{X, Y}}\max_{r_{X\mid Y}} \Bigl\{ \frac{\alpha}{1-\alpha} D(\tilde{q}_{X, Y} || p_{X}p_{Y\mid X}) \notag \\
&\qquad \qquad \qquad + \vE^{\tilde{q}_{X, Y}}\left[\log \frac{r_{X\mid Y}(X\mid Y)}{\tilde{q}_{X}(X)}\right] \Bigr\} \\ 
&\overset{(g)}{=} \max_{r_{X\mid Y}} \min_{\tilde{q}_{X, Y}}\Bigl\{ \frac{\alpha}{1-\alpha} D(\tilde{q}_{X, Y} || p_{X}p_{Y\mid X}) \notag \\
&\qquad \qquad \qquad + \vE^{\tilde{q}_{X, Y}}\left[\log \frac{r_{X\mid Y}(X\mid Y)}{\tilde{q}_{X}(X)}\right] \Bigr\} \\ 
&= \max_{r_{X\mid Y}} \min_{\tilde{q}_{X, Y}}\Bigl\{ \frac{\alpha}{1-\alpha} D(\tilde{q}_{X, Y} || \tilde{q}_{X}p_{Y\mid X}) \notag \\ 
&\qquad + \vE^{\tilde{q}_{X, Y}}\left[\log \frac{r_{X\mid Y}(X\mid Y)}{\tilde{q}_{X}(X)}\right] + D(\tilde{q}_{X} || p_{X})\Bigr\} \\
&= \max_{r_{X\mid Y}}\min_{\tilde{q}_{X}}\min_{\tilde{q}_{Y\mid X}} \Bigl\{\frac{\alpha}{1-\alpha} \vE^{\tilde{q}_{X}}\left[D(\tilde{q}_{Y\mid X}(\cdot \mid X) || \hat{q}_{Y\mid X}(\cdot \mid X))\right] \notag \\ 
& + H(\tilde{q}_{X}) + \frac{\alpha}{1-\alpha}D(\tilde{q}_{X} || p_{X})\Bigr\} \\ 
&\overset{(h)}{=}\max_{r_{X\mid Y}}\min_{\tilde{q}_{X}} \Bigl\{\frac{2\alpha-1}{1-\alpha} D(\tilde{q}_{X} || \hat{q}_{X})+\frac{2\alpha-1}{\alpha-1} \notag \\
&\times \log \sum_{x}p_{X}(x)^{\frac{\alpha}{2\alpha-1}} \left( \sum_{y} p_{Y\mid X}(y | x) r_{X\mid Y}(x | y)^{1-\frac{1}{\alpha}} \right)^{\frac{\alpha}{2\alpha-1}}\Bigr\} \label{eq:proof_end_LP} \\ 
&\overset{(i)}{=} \eqref{eq:VC_LP_MI}, 
\end{align}
where $(e)$ follows from \cite[Lemma 8]{e21080778}, 
$(f)$ follows from the variational representation of the Shannon MI using a reverse channel, 
$(g)$ follows from the minimax theorem\footnote{Notably, 
$\tilde{F}_{\alpha}^{\text{LP}}(\tilde{q}_{X, Y}, r_{X\mid Y}):=\frac{\alpha}{1-\alpha} D(\tilde{q}_{X, Y} || p_{X}p_{Y\mid X}) + \vE^{\tilde{q}_{X, Y}}\left[\log \frac{r_{X\mid Y}(X\mid Y)}{\tilde{q}_{X}(X)}\right] 
= \frac{\alpha}{1-\alpha} D(\tilde{q}_{X, Y} || \tilde{q}_{X}p_{Y\mid X}) - \frac{\alpha}{1-\alpha}\vE^{\tilde{q}_{X}}[\log p_{X}(X)] + \vE^{\tilde{q}_{X, Y}}\left[\log r_{X\mid Y}(X\mid Y)\right] + \frac{1-2\alpha}{1-\alpha}H(\tilde{q}_{X}), \alpha \in (1/2, 1)$ is concave with respect to $r_{X\mid Y}$ and is convex with respect to $\tilde{q}_{X, Y}$.}, 
and $(h), (i)$ follows from the properties of the Kullback--Leibler divergence. 
For $\alpha \in (1, \infty)$, Eq. \eqref{eq:VC_LP_MI} can be obtained by replacing $\min_{\tilde{q}_{X, Y}}$ with $\max_{\tilde{q}_{X, Y}}$ in Eqs. \eqref{eq:proof_begin_LP}--\eqref{eq:proof_end_LP} (see \cite[Lemma 8]{e21080778}).
\end{proof}

Using Lemma \ref{lem:invariance_renyi_entropy} and \ref{lem:VC_AC_LP_MI}, we prove Theorem \ref{thm:alpha_MI_diff} as follows.
\begin{proof}
From Eqs.\eqref{eq:Sibson_MI_Gallager} and \eqref{eq:Arimoto_MI_Gallager} and 1) of Lemma \ref{lem:invariance_renyi_entropy}, it follows that $I_{\alpha}^{\text{S}}(p_{X}, p_{Y\mid X}) = I_{\alpha}^{\text{A}}(p_{X_{\frac{1}{\alpha}}}, p_{Y\mid X})$.
By using 2) of Lemma \ref{lem:invariance_renyi_entropy} and Eq. \eqref{eq:Arimoto_cont_renyi_ent} and comparing with the Arimoto MI defined in \eqref{eq:Arimoto_MI}, we obtain \eqref{eq:Sibson_diff_expression}. 
Eq. \eqref{eq:AC_diff_expression} follows immediately from \eqref{eq:VC_AC_MI}. 
We prove \eqref{eq:LP_cond_renyi_ent} as follows.
\begin{align}
&I_{\alpha}^{\text{LP}}(X; Y) \notag \\
&\overset{(a)}{=} H_{\frac{\alpha}{2\alpha-1}}(X) - H_{\frac{\alpha}{2\alpha-1}}(X) + \max_{r_{X\mid Y}} \Big\{\frac{2\alpha-1}{\alpha-1}  \notag \\ 
&\times \log \sum_{x}p_{X}(x)^{\frac{\alpha}{2\alpha-1}} \left( \sum_{y} p_{Y\mid X}(y | x) r_{X\mid Y}(x | y)^{1-\frac{1}{\alpha}} \right)^{\frac{\alpha}{2\alpha-1}} \Big\} \\ 
&= H_{\frac{\alpha}{2\alpha-1}}(X) -\frac{2\alpha-1}{\alpha-1}\log \sum_{x}p_{X}(x)^{\frac{\alpha}{2\alpha-1}} + \max_{r_{X\mid Y}} \Big\{\frac{2\alpha-1}{\alpha-1}  \notag \\ 
&\times \log \sum_{x}p_{X}(x)^{\frac{\alpha}{2\alpha-1}} \left( \sum_{y} p_{Y\mid X}(y | x) r_{X\mid Y}(x | y)^{1-\frac{1}{\alpha}} \right)^{\frac{\alpha}{2\alpha-1}} \Big\} \\ 
&= H_{\frac{\alpha}{2\alpha-1}}(X) + \max_{r_{X\mid Y}} \Big\{\frac{2\alpha-1}{\alpha-1}  \notag  \\ 
&\times \log \sum_{x}\frac{p_{X}(x)^{\frac{\alpha}{2\alpha-1}}}{\sum_{x}p_{X}(x)^{\frac{\alpha}{2\alpha-1}}} \Bigl( \sum_{y} p_{Y\mid X}(y | x) r_{X\mid Y}(x | y)^{1-\frac{1}{\alpha}} \Bigr)^{\frac{\alpha}{2\alpha-1}} \Big\} \\ 
&= \eqref{eq:LP_cond_renyi_ent},
\end{align}
where $(a)$ follows from Eq. \eqref{eq:VC_LP_MI}.
\end{proof}

\section{Proof of Theorem \ref{thm:interpetations_alpha_MI}} \label{proof:interpetations_alpha_MI}
To prove Theorem \ref{thm:interpetations_alpha_MI}, we first state the following lemma. 
\begin{lemma} \label{lem:maximal_expected_score}
Given a distribution $p_{X}$, 
\begin{align}
&\max_{r} \vE^{p_{X}}[g_{\alpha}(X, r)] = \max_{r} \vE^{p_{X}}[g_{\alpha, \text{PS}}(X, r)] =  \frac{\alpha}{\alpha-1} \norm{p_{X}}_{\alpha}, \\ 
&\max_{r} \vE^{p_{X}}[g_{\alpha, \text{PW}}(X, r)] =\frac{1}{\alpha-1} \norm{p_{X}}_{\alpha}^{\alpha}.
\end{align}
\end{lemma}
\begin{proof}
It follows from \text{\cite[Lemma 1]{8804205}} and the fact that $g_{\alpha, \text{PS}}(x, r)$ and $g_{\alpha, \text{PW}}(x, r)$ are PSR.
\end{proof}
We prove Theorem \ref{thm:interpetations_alpha_MI} as follows.
\begin{proof}
Eq. \eqref{eq:Arimoto_privacy_01} has already been proven by Liao \textit{et al.}\cite[Thm 1]{8804205}. 
From Proposition \ref{prop:relationship_Arimoto_Sibson} and 1) of Lemma \ref{lem:invariance_renyi_entropy}, \eqref{eq:Sibson_privacy_01} can be obtained by replacing $p_{X}$ with $p_{X_{\frac{1}{\alpha}}}$ in \eqref{eq:Arimoto_privacy_01}.
Eq. \eqref{eq:AC_privacy_01} follows from Theorem \ref{thm:alpha_MI_diff} and the following representation of the Shannon entropy, which can be easily verified.
\begin{align}
H(X) &= \frac{\alpha}{1-\alpha}\max_{r} \log \mathbb{G}^{p_{X}} \left[r(X)^{1-\frac{1}{\alpha}}\right] \\ 
&= \frac{\alpha}{1-\alpha} \log \max_{r}\mathbb{G}^{p_{X}} \left[r(X)^{1-\frac{1}{\alpha}}\right].
\end{align}
Eq. \eqref{eq:Hayashi_privacy_01} follows from Remark \ref{rem:VC_Arimoto_Hayashi_ent} and Lemma \ref{lem:maximal_expected_score}.
Eq. \eqref{eq:LP_privacy_01} follows from Theorem \ref{thm:alpha_MI_diff}, Definitions \ref{def:power_mean} and \ref{def:gen_geometic_mean}, and the following representation of the R{\' e}nyi entropy,  which can be easily verified from Lemma \ref{lem:maximal_expected_score}.
\begin{align}
H_{\alpha}(X) &= \frac{\alpha}{1-\alpha}\log \max_{r} \vE^{p_{X}} \left[r(X)^{1-\frac{1}{\alpha}}\right] \\ 
&= -\log \max_{r} \mathbb{M}_{1-\frac{1}{\alpha}}^{p_{X}}[r(X)].
\end{align}
Eqs. \eqref{eq:Arimoto_privacy_02}, \eqref{eq:Sibson_privacy_02}, and \eqref{eq:AC_privacy_02} follow from Lemma \ref{lem:maximal_expected_score}. 
The remainder follows from Definitions \ref{def:power_mean} and \ref{def:gen_geometic_mean} and Lemma \ref{lem:equivalent_power_gen_geometic_mean}. 
\end{proof}


\end{document}